%% file: MTMC.tex
\documentclass[12pt]{article}
\usepackage{amsmath}
\usepackage{graphicx,psfrag,epsf}
\usepackage{enumerate}
\usepackage[square,numbers]{natbib}
\usepackage{url} 

\usepackage[OT1]{fontenc}

\usepackage[english]{babel}

\usepackage[utf8]{inputenc}

\usepackage[sc]{mathpazo}

\usepackage{amsmath,amssymb,amsfonts,mathrsfs}

\usepackage[amsmath,thmmarks]{ntheorem}

\usepackage{graphicx}

\usepackage{soul}

\usepackage{pdfpages}

\input{theoremsetup}

\input{extrapackages}

\usepackage[linkcolor=black,colorlinks=true,citecolor=black,filecolor=black]{hyperref}

\usepackage{algorithm}
\usepackage[noend]{algpseudocode}
\def\BState{\State\hskip-\ALG@thistlm}
\usepackage{tabularx,ragged2e,booktabs,caption}
\usepackage{comment}
\usepackage{etoolbox}
\makeatletter
\usepackage{lipsum}
\allowdisplaybreaks

\newcommand{\blind}{1}

\begin{document}

\def\spacingset#1{\renewcommand{\baselinestretch}%
{#1}\small\normalsize} \spacingset{1}


\if1\blind
{
  \title{\bf Moving Target Monte Carlo}
  \author{Haoyun Ying\hspace{.2cm}\\
    Department of Mathematics, ETH Zurich\\
    and \\
    Keheng Mao \\
    Institute of Geophysics, ETH Zurich\\
    and\\ 
    Klaus Mosegaard \\
    Niels Bohr Institute, University of Copenhagen}
  \date{}
  \maketitle
} \fi

\if0\blind
{
  \bigskip
  \bigskip
  \bigskip
  \begin{center}
    {\LARGE\bf Moving Target Monte Carlo}
\end{center}
  \medskip
} \fi

\begin{abstract}
\input{abstract}
\end{abstract}

\noindent%
\vfill

\newpage
\begin{center}

\textbf{Acknowledgement}

\end{center}

\bigskip

\textit{I would like to thank Prof. Dr. Van de Geer of Department of Mathematics at ETH Zurich for her careful readings of this manuscript and fruitful discussions during the project.}

\bigskip

\textit{I would also like to acknowledge Prof. Dr. Andereas Fichtner of Institute of Geophysics at ETH Zurich for joining on this project at the very early stage. I am gratefully indebted to his very valuable ideas.}

\bigskip

\textit{Finally, I must express my very profound gratitude to M. Ying and X. Fei for their unfailing support and continuous encouragement throughout the process of researching and writing this paper.}

\bigskip

\textit{This accomplishment would not have been possible without them.}

\newpage
\tableofcontents

\spacingset{1.45} 

\newpage

\section{Introduction}

\input{introduction.tex}

\input{sections}
\bibliographystyle{unsrt}
\bibliography{MTMC.bib}
\end{document}

%% file: theoremsetup.tex
\numberwithin{equation}{section}

\newtheorem{theorem}{Theorem}[section]

\newtheorem{corollary}[theorem]{Corollary}
\newtheorem{definition}[theorem]{Definition}
\newtheorem{lemma}[theorem]{Lemma}
\newtheorem{proposition}[theorem]{Proposition}

\theoremstyle{nonumberplain}
\theorembodyfont{\normalfont}
\theoremsymbol{\ensuremath{\square}}
\newtheorem{proof}{Proof}

%% file: extrapackages.tex
\usepackage{varioref}

\usepackage{datetime}

\usepackage{mathtools}

\usepackage[h]{esvect}

\usepackage{array}

\usepackage{listings}
\usepackage{booktabs}

%% file: abstract.tex
The Markov Chain Monte Carlo (MCMC) methods are popular when considering sampling from a high-dimensional random variable $\mathbf{x}$ with possibly unnormalised probability density $p$ and observed data $\mathbf{d}$. However, MCMC requires evaluating the posterior distribution $p(\mathbf{x}|\mathbf{d})$ of the proposed candidate $\mathbf{x}$ at each iteration when constructing the acceptance rate. This is costly when such evaluations are intractable. In this paper, we introduce a new non-Markovian sampling algorithm called Moving Target Monte Carlo (MTMC). The acceptance rate at $n$-th iteration is constructed using an iteratively updated approximation of the posterior distribution $a_n(\mathbf{x})$ instead of $p(\mathbf{x}|\mathbf{d})$. The true value of the posterior $p(\mathbf{x}|\mathbf{d})$ is only calculated if the candidate $\mathbf{x}$ is accepted. The approximation $a_n$ utilises these evaluations and converges to $p$ as $n \rightarrow \infty$. A proof of convergence and estimation of convergence rate in different situations are given.

%% file: introduction.tex
We consider the problem of sampling from a high-dimensional random variable $\mathbf{x}$ with possibly unnormalised probability distribution (target distribution) $p$. The Markov Chain Monte Carlo (MCMC) methods such as Metropolis-Hastings (MH) \cite{met} \cite{met2} are popular approaches which sample from posterior distribution $p(\mathbf{x}|\mathbf{d})$ conditioning on the observed data $\mathbf{d}$ utilising the Bayesian inference framework and guarantee the convergence of the sampling distribution to $p$ using the Markov chain properties. In practice where the evaluation of posterior distribution $p(\mathbf{x}|\mathbf{d})$ is intractable, the MCMC method could be expensive since this is required at each iteration when deciding whether the candidate selected by the proposal distribution is accepted. 

Over the last few years, there are researches in physics-related fields working on raising the acceptance rate of the proposed candidates so that most of the evaluations are not wasted. One could approximate $p$ based on its local structure and generate a 'promoted' proposal distribution similar to $p$, this would result a higher acceptance rate (J. Christen and C. Fox 2005 \cite{fox2}). In this way, the candidates are 'pre-selected' before the evaluations thus no evaluations are wasted. 

Employing approximations of the posterior $p$ has proved itself useful in reducing computations involving $p$. In this paper, we take the idea of using approximation one step further introducing the Moving Target Monte Carlo (MTMC) method: At step $n$, the $n$-th approximation $a_n$ of $p$ is used to construct the acceptance rate deciding if the $n$-th candidate shall be accepted. The posterior is only evaluated if the candidate is accepted so no evaluations are wasted. Furthermore, this evaluation is used to update the approximation to $a_{n+1}$ for the decision at step $n+1$. 

The updating of the approximation is better than fixing one approximation in many ways. First of all, since the approximation $a_n$ converges to $p$ with every candidate accepted and the posterior of it calculated, $a_n$ performs better than any other approximation in the long run. Also, the limitation of fixed approximations lies in the necessity of having a lot of pre-knowledge about the target distribution in order to be able to choose a reasonable approximation whereas MTMC needs no prior information, which means it can even start with an uniform distribution as approximation and still manage to sample with the target distribution in the end. Therefore, the application of MTMC goes beyond the physics-related fields.

On the other hand, there are yet no researches considering updating the approximation probably because the resulting chain is not Markovian and one cannot use Markov chain properties to prove convergence. Indeed, the approximation utilises all historical information and thus depends on all the sample points accepted. In this paper, we prove the convergence of MTMC using some ideas from the recently developed adaptive MCMC which also has a non-Markovian transition kernel.

In Section 2, we introduce the Markov Chain Monte Carlo based on the Metropolis-Hastings algorithm since our algorithm has a similar structure to the Metropolis-Hastings. Subsection 3.1 describes the Moving Target Monte Carlo based on a simple example using the Nearest Neighbour distribution as the approximation method in Section 3.2. For the convergence proof in Section 4, we first review the coupling argument in Subsection 4.1 and give our main result--the proof of convergence--based on it in Subsection 4.2. The main result has two constraints. We proceed by showing that the MTMC chain meets these two constraints. Finally, the estimation of the convergence rate is given in Section 5. In Subsection 5.1, we reduce the case by using an independent proposal and bound the variation distance using the eigenvalue analysis. Subsection 5.2 discusses the general case and again uses the coupling argument. 

%% file: sections.tex
\section{Markov Chain Monte Carlo(MCMC)}

\bigskip\noindent

\bigskip\noindent
\subsection{Bayesian Inference}
\bigskip\noindent
For a data point $\mathbf{d}\in D$ and a parameter $\mathbf{x} \in \mathcal{H}$, where $D$ and $\mathcal{H}$ denote the \emph{data space} and \emph{model space} respectively \cite{Tarantola}, 
the Bayesian inference predicts \emph{posterior distribution} $p(\mathbf{x}|\mathbf{d})$ (distribution of the parameters conditional on observed data) based on \emph{likelihood function}
$p(\mathbf{d}|\mathbf{x})$ (distribution of the observed data conditional on its parameters) and \emph{initial prior distribution} $p_0(\mathbf{x})$ (distribution of the parameters before observing data):
\begin{equation}
\label{probdens}
p(\mathbf{x}|\mathbf{d}) = k^{-1} p(\mathbf{d}|\mathbf{x})p_0(\mathbf{x}).
\end{equation}
The \emph{marginal likelihood}, $k=\int_{\mathcal{H}} p(\mathbf{d}|\mathbf{x})p_0(\mathbf{x}) d \mathbf{x}$, normalises the equation so that (\ref{probdens}) is indeed a probability density function.
\bigskip\noindent
\subsection{Markov Chain Monte Carlo (MCMC) and Metropolis-
Hastings (MH)}
\bigskip\noindent
Given prior information $p_0(\mathbf{x})$ and likelihood function $p(\mathbf{d}|\mathbf{x})$, 
one obtains posterior distribution $p(\mathbf{x}|\mathbf{d})$ with Bayes' rule only if the problem has an analytical solution. 
However, this often requires obtaining intractable inverse expression of forward relations. 
For instance, with a mass distribution and the location of the observer, it is trivial to obtain the corresponding gravitational force by applying the forward relation using Newton’s law of universal gravitation. 
However, the other way around could be difficult: Given the gravitational forces observed from several locations, finding the right mass distribution by applying the inverse relation can be challenging. 
In such case, we may proceed with an approximated posterior  $p(\mathbf{x}|\mathbf{d})$ based on \emph{Markov chain Monte Carlo (MCMC) methods}~\cite{Klaus}. 

MCMC algorithms are Markov chains where the probability of a transition from $\mathbf{x}$ to $\mathbf{y}$ in $\mathcal{H}$ is conditionally independent of the information from previously visited points. This is described by the \emph{transition
probability distributions} $P(\mathbf{x},\mathbf{y})$.

Consider an MCMC with transition probabilities
$P(\mathbf{x},\mathbf{y})$ with the sample distribution at $N$-th step $p_N$. If $p_N$ converges to some distribution $p$, then $p$ is a \emph{stationary} (or \emph{equilibrium}) distribution for $P(\mathbf{x},\mathbf{y})$. If the chain equilibrates to $p$ independently of the initial distribution $p_{0}$, then $p$ is the \emph{unique} stationary distribution for $P(\mathbf{x},\mathbf{y})$.

It is thus natural to consider designing a Markov chain with posterior distribution $p(\mathbf{x}|\mathbf{d})$ as the stationary distribution to draw independent samples using Monte Carlo method. The name MCMC comes from this idea. 

The sampling distribution can be used to evaluate the problem instead of $p(\mathbf{x}|\mathbf{d})$.

To illustrate the idea behind this, we consider one of the most popular MCMC methods, the \emph{Metropolis-Hastings (MH) algorithm} \cite{met} \cite{met2}. 

\begin{algorithm}[H]
\caption{Metropolis-Hastings Algorithm \cite{Yildirim}}
\label{Alg1}
\begin{algorithmic}[1]
\State Initialise the algorithm with $\mathbf{x}^{0} $
\For {iteration $n = 0,1,\dots,N-2$}
    \State Propose: $\mathbf{x}^{n+1,*} \sim Q(\mathbf{x}^{n},\mathbf{x}^{n+1,*})$
    \State Evaluate $p(\mathbf{x}^{n+1,*}|\mathbf{d})$; 
    \State Acceptance Probability:
    
    $\alpha_1(\mathbf{x}^{n},\mathbf{x}^{n+1,*}) =
    \min \Big(1,\frac{p(\mathbf{x}^{n+1,*}|\mathbf{d})}{p(\mathbf{x}^{n}|\mathbf{d})}\frac{Q(\mathbf{x}^{n+1,*},\mathbf{x}^{n})}{Q(\mathbf{x}^{n},\mathbf{x}^{n+1,*})}\Big)$
    \If {$u<\alpha_1(\mathbf{x}^{n},\mathbf{x}^{n+1,*})$ with $u \sim$ Uniform$(0,1)$}
        \State Accept the proposal: $\mathbf{x}^{n+1} \gets \mathbf{x}^{n+1,*}$
    \Else
        \State Reject the proposal: $\mathbf{x}^{n+1} \gets \mathbf{x}^{n}$
    \EndIf
\EndFor
\Return The sequence of $N$ points $\mathbf{x}^{0}, \mathbf{x}^{1}, \dots, \mathbf{x}^{N-1}$
\end{algorithmic}
\end{algorithm}

MH generates samples of a given distribution iteratively based on the following inputs: 
\begin{enumerate}
    \item The data $\mathbf{d}\in D$;
    \item The \emph{sample size} $N$ defined as the size of the samples one wants in the end;
    \item The target distribution $p$ with which the value $p(\mathbf{x}|\mathbf{d})$ is computable given $\mathbf{x}$ and $\mathbf{d}$;
    \item The step-wise \emph{proposal distribution} $Q(\mathbf{x}^{n},\mathbf{x}^{n+1,*})$ defined as the probability with which the Markov chain proposes $\mathbf{x}^{n+1,*}$ from current position $\mathbf{x}^{n}$.
\end{enumerate}

The algorithm is initialised by choosing the first sample $\mathbf{x}_0$ arbitrarily. Suppose a sequence of $n$ points $\mathbf{x}^{0}, \mathbf{x}^{1}, \dots, \mathbf{x}^{n-1}$ are previously visited by the algorithm.

The main loop of the algorithm comprises three components:
\begin{enumerate}
    \item At each step $n \geq 0$, 
propose a sample $\mathbf{x}^{n+1,*}$ with the proposal distribution $Q(\mathbf{x}^{n},\mathbf{x}^{n+1,*})$ conditioning on the current sample $\mathbf{x}^{n}$.
    \item Compute the acceptance ratio $\alpha_1(\mathbf{x}^{n},\mathbf{x}^{n+1,*})$ with given invariant target distribution $p(\mathbf{x}|\mathbf{d})$ and proposal distribution $Q(\mathbf{x}^{n},\mathbf{x}^{n+1,*})$: 
\begin{equation}
\label{alpha1}
\alpha_1(\mathbf{x}^{n},\mathbf{x}^{n+1,*}) =
    \min \bigg(1,\frac{p(\mathbf{x}^{n+1,*}|\mathbf{d})}{p(\mathbf{x}^{n}|\mathbf{d})}\frac{Q(\mathbf{x}^{n+1,*},\mathbf{x}^{n})}{Q(\mathbf{x}^{n},\mathbf{x}^{n+1,*})}\bigg).
\end{equation}
    \item Generate a random number $u \in [0,1]$. 
    
    Accept the candidate if $u<\alpha_1$; 
    
    Reject if otherwise and the chain remains at $\mathbf{x}^{(n)}$.
\end{enumerate}
Finally, the algorithm returns a sequence of $N$ samples points $\mathbf{x}^{0}, \mathbf{x}^{1}, \dots, \mathbf{x}^{N-1}$. The sampling distribution converges to $p(\mathbf{x}|\mathbf{d})$.

$Remark.$ In particular, one can use \emph{symmetric proposal distribution} such as random walks: 
\begin{equation}\nonumber
    Q(\mathbf{x},\mathbf{y})=Q(\mathbf{y},\mathbf{x}).
\end{equation}
Then the \emph{acceptance ratio} $\alpha_1$ is simplified to 
\begin{equation*}
    \alpha_1(\mathbf{x}^{n},\mathbf{x}^{n+1,*}) =
    \min \bigg(1,\frac{p(\mathbf{x}^{n+1,*}|\mathbf{d})}{p(\mathbf{x}^{n}|\mathbf{d})}\bigg)
\end{equation*}

\begin{figure}[H]
\includegraphics[width=16cm,height=20cm,keepaspectratio]{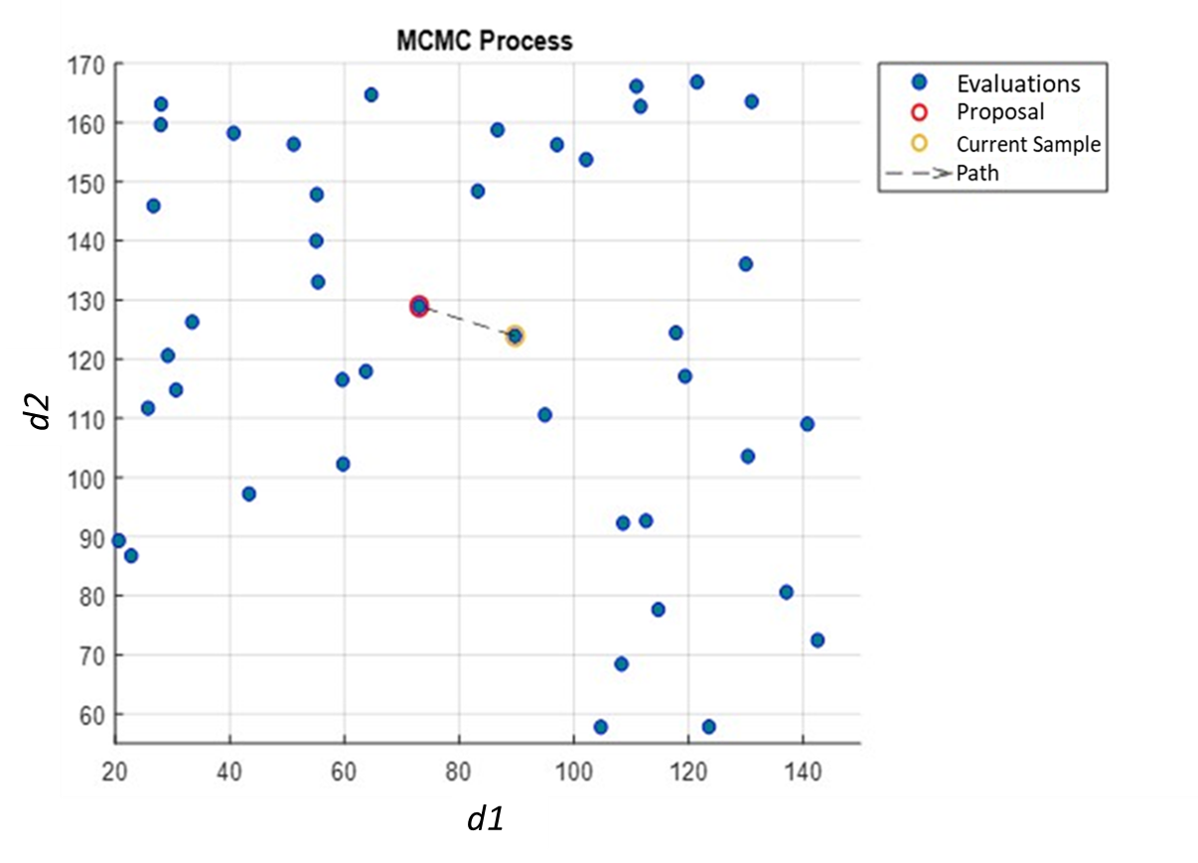}
\caption{Visualization of one iteration of Algorithm~\ref{Alg1} on 2-dimensional Euclidean space. Let $\mathbf{x}^{n}=(d^{n}_1,d^{n}_2)$ denote the 2-dimensional parameter. 
The blue dots indicate all evaluated points. The red and the yellow circles indicate the proposed sample $\mathbf{x}^{n+1,*}$ the current sample $\mathbf{x}^{n}$ respectively. 
The path indicates a potential transition. If the acceptance ratio $\alpha_1$ is bigger than the random generated number $u$, the proposed sample point is accepted. 
On the other hand, if the proposed sample is rejected, the point itself and its evaluation are discarded out of memory. It is to notice that all other former evaluations are irrelevant to the decision at the $n$-th step.}
\label{fig1}
\end{figure}

\bigskip\noindent

\section{Moving Target Monte Carlo (MTMC)}

\bigskip\noindent

Although Metropolis-Hastings chain samples fast and precise, it has two drawbacks in problems where evaluating posterior is expensive:

\begin{enumerate}
    \item At each step, MH evaluates the posterior of the candidate using both proposal distribution and the given target distribution before making decisions. 
    Thus if the algorithm rejects the proposed sample, the corresponding evaluation is discarded. The time and the cost for this evaluation are wasted;
    \item Only two evaluations are involved when calculating the acceptance probability, whereas the information about other evaluated points which also contains knowledge about the posterior distribution is not utilized in the decision-making process.
\end{enumerate}

In the previous section, we emphasised that the target distribution $p(\mathbf{x}|\mathbf{d})$ depends on data set $D$. This indicates that MH has to calculate $p(\mathbf{x}^{n+1,*}|\mathbf{d})$ using the (possibly a giant amount of) data in every iteration. 
Such calculations can be extremely intractable in many real-world cases as described in the Introduction.

We introduce a new algorithm alleviating the computation by evaluating proposals with an approximation based on historical information of the visited points. This minimises the amount of evaluations using the true distribution $p(\mathbf{x}|\mathbf{d})$ needed.
\bigskip\noindent
\subsection{Description of the method}
\bigskip\noindent
In the following, we describe a non-Markovian algorithm sampling a given target distribution $p(\mathbf{x}|\mathbf{d})$ using information from all earlier sample points in the process. The algorithm proceeds using a rough approximation $a_n(\mathbf{x})$ of $p(\mathbf{x}|\mathbf{d})$ at $n$-th iteration and successively improves the approximation so that it converges to $p$. The true distribution $p(\mathbf{x}|\mathbf{d})$ is only calculated if the sample point $\mathbf{x}$ is accepted based on the approximation distribution $a_n$. The evaluation of the posterior of $\mathbf{x}$ afterwards is also used to update the approximation $a_n$ to a better approximation $a_{n+1}$ of $p(\mathbf{x}|\mathbf{d})$. 

\begin{algorithm}[H]
\caption{Moving Target MC Algorithm}\label{Alg2}
\begin{algorithmic}[1]
\State Initialise the algorithm with $\mathbf{x}^{0}$
\For {iteration $n = 0,1,\dots,N-2$}
    \State Propose: $\mathbf{x}^{n+1,*} \sim Q(\mathbf{x}^{n},\mathbf{x}^{n+1,*})$
    \State Acceptance Probability:
    $\alpha_2(\mathbf{x}^{n},\mathbf{x}^{n+1,*}) = \min \Big(1,\frac{a_{n}(\mathbf{x}^{n+1,*})}{a_{n}(\mathbf{x}^{n})}\frac{Q(\mathbf{x}^{n+1,*},\mathbf{x}^{n})}{Q(\mathbf{x}^{n},\mathbf{x}^{n+1,*})}\Big)$
    \If {$u<\alpha_2(\mathbf{x}^{n},\mathbf{x}^{n+1,*}), u \sim$ Uniform$(u;0,1)$}
        \State Accept the proposal: $\mathbf{x}^{n+1} \gets \mathbf{x}^{n+1,*}$; Evaluate $p(\mathbf{x}^{n+1}|\mathbf{d})$;
        \State Update $a_{n}(\mathbf{\cdot})$ to $a_{n+1}(\mathbf{\cdot})$ using $a_{n+1}(\mathbf{x}^{n+1}) = p(\mathbf{x}^{n+1}|\mathbf{d})$
    \Else
        \State Reject the proposal: $\mathbf{x}^{n+1} \gets \mathbf{x}^{n}$; $a_{n+1}(\mathbf{\cdot}) \gets a_{n}(\mathbf{\cdot})$
    \EndIf
\EndFor
\Return The sequence of $N$ points $\mathbf{x}^{0}, \mathbf{x}^{1}, \dots, \mathbf{x}^{N-1}$
\end{algorithmic}
\end{algorithm}

MTMC generates samples of a given distribution iteratively based on the following inputs:
\begin{enumerate}
    \item The data $\mathbf{d} \in D$;
    \item The sample size $N$;
    \item The target distribution $p$ with which the value $p(\mathbf{x}|\mathbf{d})$ is computable given $\mathbf{x}$ and $\mathbf{d}$;
    \item The approximation method $(a_n)_{n \in A}$ where $A$ is the approximation index;
    \item The step-wise proposal distribution $Q(\mathbf{x}^{n},\mathbf{x}^{n+1,*})$.
\end{enumerate}

The algorithm is initialised by choosing the first sample $\mathbf{x}_0$ arbitrarily. Suppose a sequence of $n$ points $\mathbf{x}^{0}, \mathbf{x}^{1}, \dots, \mathbf{x}^{n-1}$ are previously visited by the algorithm. 

Then the corresponding sequence of distributions $\{a_n(\mathbf{x})\}_{n \in \mathbb{N}}$ satisfies:
\begin{equation}
\label{req1}
a_n(\mathbf{x}^k) = p(\mathbf{x}^{k}|\mathbf{d}) \text{ for } k = 0, \dots, n-1  \ ,
\end{equation}
and we assume for each point $\mathbf{x} \in \mathcal{H}$
\begin{equation}
\label{req2}
\lim_{n \rightarrow \infty} \sup_{\mathbf{x}} \lvert a_{n} (\mathbf{x}) - p(\mathbf{x}|\mathbf{d}) \rvert \rightarrow 0
\end{equation}
and
\begin{equation}
\label{req3}
    \sup_x \lvert a_{n+1} (\mathbf{x}) - a_{n}(\mathbf{x}) \rvert \rightarrow 0 \text{ in probability},\text{for} \ n \rightarrow \infty.
\end{equation}
We assume that any point in $\mathcal{H}$ has non-zero probability of being visited after sufficient steps.

The main loop of the algorithm comprises three components:
\begin{enumerate}
    \item At each step $n \geq 0$, propose a sample $\mathbf{x}^{n+1,*}$ with the proposal distribution $Q(\mathbf{x}^{n},\mathbf{x}^{n+1,*})$ conditioning on the current sample $\mathbf{x}^{n}$.
    \item Compute the acceptance ratio $\alpha_2(\mathbf{x}^{n},\mathbf{x}^{n+1,*})$ with the current approximation distribution $a_n(\mathbf{x})$ and proposal distribution $Q(\mathbf{x}^{n},\mathbf{x}^{n+1,*})$: 
    \begin{equation}
    \label{alpha2}
    \alpha_2(\mathbf{x}^{n},\mathbf{x}^{n+1,*}) = \min \bigg(1,\frac{a_{n}(\mathbf{x}^{n+1,*})}{a_{n}(\mathbf{x}^{n})}\frac{Q(\mathbf{x}^{n+1,*},\mathbf{x}^{n})}{Q(\mathbf{x}^{n},\mathbf{x}^{n+1,*})}\bigg).
    \end{equation}
    \item Generate a random number $u \in [0,1]$. 
    
    Accept the candidate if $u<\alpha_2$: Calculate and save the value of $p(\mathbf{x}^{n+1}|\mathbf{d})$. Update the approximation to $a_{n+1}$ using $a_{n+1}(\mathbf{x}^{n+1}) = p(\mathbf{x}^{n+1}|\mathbf{d})$;
    
    Reject the candidate if $u<\alpha_2$: The chain remains at $\mathbf{x}^{(n)}$. No evaluation is made and $a_{n+1}=a_n$.
\end{enumerate}

Finally, the algorithm returns a sequence of $N$ sample points $\mathbf{x}^{0}, \mathbf{x}^{1}, \dots, \mathbf{x}^{N-1}$. The distribution of the sample converges to $p$ when $N \rightarrow \infty$.

$Remark.$ Similar to Algorithm~\ref{Alg1}, symmetric proposal distributions can be applied to $\alpha_2$ in order to simplify the argument inside $\min$.

Note that the acceptance probability (\ref{alpha2}) only depends on the current approximation $a_n$ of $p$ and the evaluation of $p(\mathbf{x}|\mathbf{d})$ is required only in accepted points. In this lies the potential advantage of the algorithm: When evaluation of $a_n(\mathbf{x})$ is much less computationally intensive than the computation of $p(\mathbf{x}|\mathbf{d})$, there is a significant gain in the computational workload.

In principle, $a_{n}$ can be any approximation with (\ref{req1}), (\ref{req2}) and (\ref{req3}) which guarantee the convergence to the true value when the sample size is large enough (see proof in the later sections). However, the rate of convergence depends on the choice of how one approximates $p$ (see proof in the later sections). 

\bigskip\noindent
\subsection{Example: Nearest Neighbour Approximation} 
\bigskip\noindent
To illustrate MTMC, let us consider a simple approximation method, the \emph{nearest neighbour constant interpolation}.

Let $\mathbf{x}^{n}_{I}$, $\mathbf{x}^{n+1,*}_{I}$ be the nearest point of $\mathbf{x}^{n}$, $\mathbf{x}^{n+1,*}$ respectively among all evaluated points $\{\mathbf{x}_i\}_{i \in I}$, where $I=\{0,\dots,t\}$ with $t$ the total number of evaluations and also the number of accepted points (It is possible that $t \neq n$ since rejection gives repeated sample points). Furthermore, since the nearest neighbour $\mathbf{x}^{n}_{I}$ to $\mathbf{x}^{n}$ among $\{\mathbf{x}_{i}\}$ is $\mathbf{x}^{n}$ itself, we rewrite (\ref{alpha2}) as
\begin{equation}
\label{eq4}
\alpha_2(\mathbf{x}^{n}, \mathbf{x}^{n+1,*}) = \min \bigg(1,\frac{p(\mathbf{x}^{n+1,*}_{I}|\mathbf{d})}{p(\mathbf{x}^{n}|\mathbf{d})}\frac{Q(\mathbf{x}^{n+1,*},\mathbf{x}^{n})}{Q(\mathbf{x}^{n},\mathbf{x}^{n+1,*})}\bigg).
\end{equation}
since the approximation $a_{n}(\mathbf{x}^{n+1,*})$ at $\mathbf{x}^{n+1,*}$ equals the value of the target distribution at its nearest neighbour $p(\mathbf{x}^{n+1,*}_{I}|\mathbf{d})$ whereas $a_{n}(\mathbf{x}^{n})$ adopts its true value $p(\mathbf{x}^{n}|\mathbf{d})$. 

\begin{definition}
Let $\{\mathbf{x}_{i}\}_{i = 0,\dots,t}$ be a set of points in d-dimensional metric space, $2\leq t \leq \infty$, $\mathbf{x}_{i} \neq \mathbf{x}_{j}$ for $i\neq j$. The Voronoi cell of $\mathbf{x}_{i}$ is~\cite{Sambridge}
\begin{equation*}
A_i=\{\mathbf{x}\big| |\mathbf{x}-\mathbf{x}_{i}|\leq |\mathbf{x}-\mathbf{x}_{j}| \text{ , for } j\neq i, \hspace{1mm} j=0,...,t\}.
\end{equation*}
\end{definition}

Figure~\ref{fig2} visualises the above described process with Voronoi cells. The posterior of $\mathbf{x}^{n+1,*}$ is not calculated to evaluate the proposed sample.
Instead, $\mathbf{x}^{n+1,*}$ adopts the value of its nearest evaluated neighbour (red dot). The current sample uses its true value (yellow dot). The simplicity of this interpolation method comes from the fact that only one extra sample $\mathbf{x}^{n+1,*}_{I}$ is involved in the interpolation process. 
Other interpolation methods such as inverse distance interpolation or k-nearest neighbour interpolation may quest multiple points in one step.
The value of all samples inside one Voronoi cell is set to be the same. One could consider Voronoi cell as the nearest neighbour region of a given point whose value is already calculated.

Let $V_{i}(\mathbf{x})$ be the basis function of Voronoi cells
\begin{equation}\nonumber
V_{i}(\mathbf{x}) = \begin{cases}
1, & \text{if $\mathbf{x} \in A_i$}.\\
0, & \text{else}.
\end{cases}
\end{equation}
The approximation using Voronoi cell at time $n$ is
\begin{equation}\nonumber
a_{n}(\mathbf{x}) = \hat{k}^{-1}_n\sum^t_{i=0}p(\mathbf{x}_{i}|\mathbf{d})V_{i}(\mathbf{x}).
\end{equation}
with $\hat{k}_n =\int_{\mathcal{H}} a_{n}(\mathbf{x}) d \mathbf{x}$ the normalisation constant. 

\begin{figure}[H]
\includegraphics[width=16cm,height=20cm,keepaspectratio]{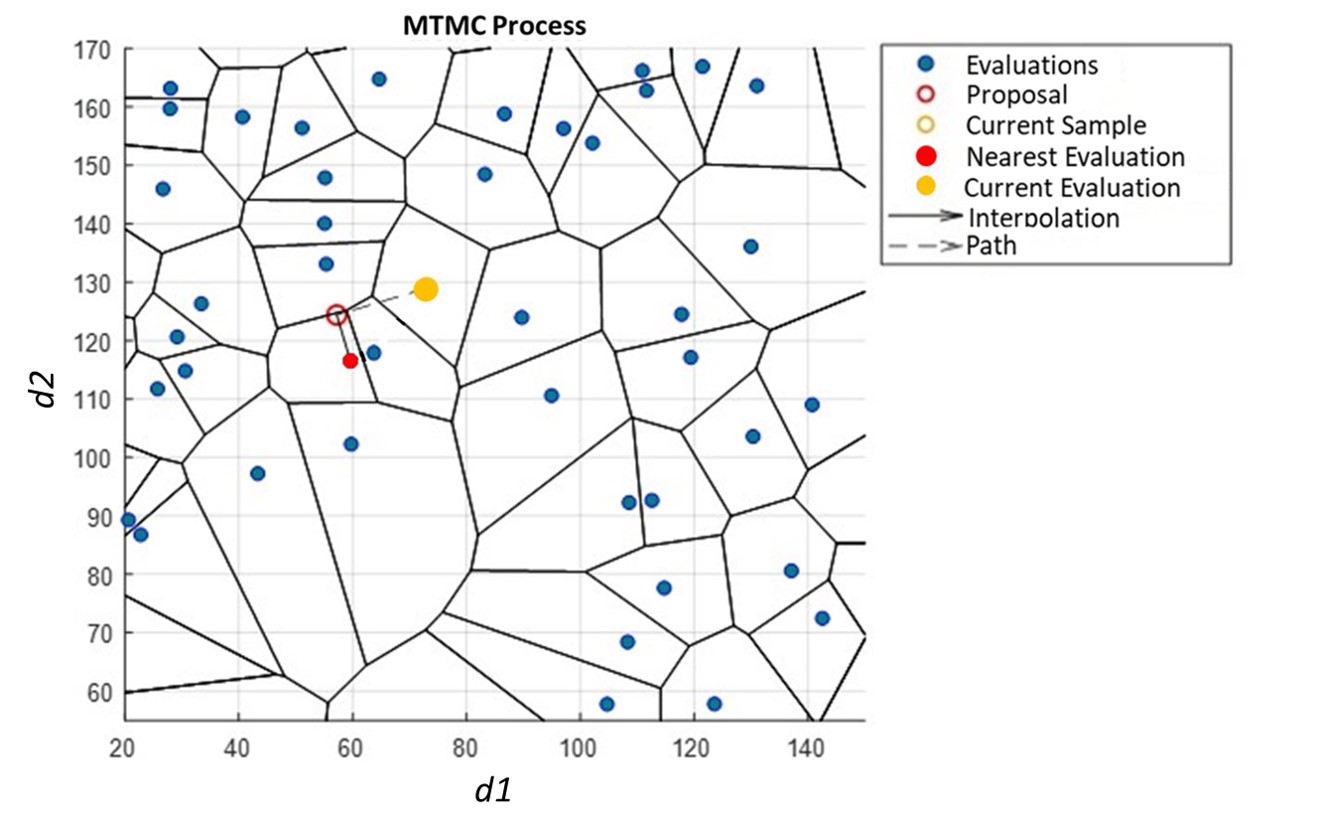}
\caption{Visualisation of one iteration of Algorithm~\ref{Alg2}. 
The setup resembles Figure~\ref{fig1}. The solid black lines partition the space into Voronoi cells centered at evaluated points (blue dots) with true values. 
The solid arrow represents the interpolation from the nearest neighbor (red dot) to the proposed sample (red circle). 
While the current sample (yellow circle) overlaps with the current evaluation (yellow dot).}
\label{fig2}
\end{figure}

Instead of sampling with the true distribution $p$ and have stationary distribution equal to $p$ at each step, MTMC samples with the approximated distributions converging to $p$. At $n$-th step, the stationary distribution is $a_{n}$ (see proof in the later sections). The stationary distribution $a_n$ evolves with time $n$ and converges to $p$ along with the sampling distribution, thus 'Moving target Monte Carlo'. 

\bigskip
\section{Proof of Convergence}
\bigskip
\subsection{The Coupling Method}
In this section, we use the classic coupling technique to prove that the MTMC chain is ergodic. The coupling technique is widely used in the convergence proof of Markov chains. However, the Markov property is not a necessary condition for the coupling argument.  

First, we shall review how to construct a coupling argument \cite{coupling}.

The intuition of coupling is the following. We have two chains $(X_k)$ and $(Y_k)$ starting in the distribution $p_0$ and $p$ accordingly and proceeds both with the same transition probability. After some \emph{stopping time} $T$, the two chains meet and the sampling outputs become equal. If we set $p$ to be the stationary distribution and $p_0$ to be the start distribution, using the \emph{coupling inequality} we could bound the \emph{variation distance} to the stationary distribution for our chain. 

First, let us define the variation distance between two probability measures.

\begin{definition}
\label{vardis}
Given probability measures $m_1$ and $m_2$ on the same measurable space $(\mathcal{H},\sigma(\mathcal{H}))$, the total variation distance between them is defined as
\begin{equation*}
    \lVert m_1-m_2 \rVert = \sup_B |m_1(B)-m_2(B)|
\end{equation*}
where the supremum is taken over all measurable subsets $B$ in $\sigma(\mathcal{H})$.
\end{definition}

$Notation.$ The probability of some event $E$ is denoted as $Prob(E)$. Similarly, the joint probability distribution of the events $E_1,\dots,E_n$ is denoted as $Prob(E_1,\dots,E_n)$. The conditional distribution of $E_1,\dots,E_n$ given $F_1,\dots,F_n$ is $Prob(E_1,\dots,E_n|F_1,\dots,F_n)$. We use $\mathcal{L}(X_k)(\cdot)$ to denote the probability distribution of the $(k+1)$-th iteration given the $k$-th iteration $X_k$. For transition probability we use $\sim$. For instance, $X_{k+1} \sim P(X_k,\cdot)$ means that the transition probability from $X_k$ to $X_{k+1}$ follows the transition kernel $P$. 

Given a chain on state space $\mathcal{H}$ with stationary distribution $p$, initial distribution $p_0$ and transition probability $P$, suppose we can find a 'joint' chain $(X_k,Y_k)$ on $\mathcal{H} \times \mathcal{H}$:
\begin{enumerate}
    \item $X_{k+1} \sim P(X_k,\cdot)$;
    \item $Y_{k+1} \sim P(Y_k,\cdot)$;
    \item $Prob(X_0=\cdot)=p_0(\cdot)$;
    \item $Prob(Y_0=\cdot)=p(\cdot)$.
\end{enumerate}

As the first two conditions state that both of these chains follow the same transition probability, the fourth condition implies $Prob(Y_k=\cdot)=p(\cdot)$ for all $k$. If there is a time $T$ such that $X_k=Y_k$ for all $k>T$, we call $(X_k,Y_k)$ a coupling. Now, we introduce the coupling inequality.

\begin{theorem}(Coupling Inequality)
\label{thm2.2}
The variation distance defined as in Definition \ref{vardis} between $\mathcal{L}(X_k)$ and $p$ is bounded above by the probability that $T>k$:
\begin{equation*}
    \rVert \mathcal{L}(X_k) - p(\cdot) \rVert \leq Prob(X_k \neq Y_k) \leq Prob(T>k) 
\end{equation*}
\end{theorem}

\begin{proof}
For any subset $B \in \sigma(\mathcal{H})$:
\begin{equation*}
\begin{split}
    |\mathcal{L}(X_k)(B)-p(B)| & = |Prob(X_k \in B)-Prob(Y_k \in B)|\\
    &=|Prob(X_k \in B, X_k = Y_k) + Prob(X_k \in B, X_k \neq Y_k)\\
    &-Prob(Y_k \in B, X_k = Y_k) - Prob(Y_k \in B, X_k \neq Y_k)|\\
    &=|Prob(X_k \in B, X_k \neq Y_k) - Prob(Y_k \in B, X_k \neq Y_k)|\\
    &\leq Prob(X_k \neq Y_k)\\
    &\leq Prob(T>k)
\end{split}
\end{equation*}
Thus,
\begin{equation*}
    \rVert \mathcal{L}(X_k) - p(\cdot) \rVert = \sup_B |\mathcal{L}(X_k) - p(\cdot)|
\end{equation*}
gives the desired inequality.
\end{proof}

We now give a general result that is used later constructing coupling chains in the proof of convergence.
The proof follows Roberts' and Rosenthal's \cite{couplingconstruct}. 

\begin{theorem}(Roberts and Rosenthal)
\label{thm2.3}
Given probability measures $m_1$ and $m_2$ on $(\mathcal{H},\sigma(\mathcal{H}))$, their total variation distance as in Definition \ref{vardis}, one can construct two chains $X_n$ and $Y_n$ with $Prob(X_n \neq Y_n) = \lVert m_1 -m_2 \rVert$.
\end{theorem}
\begin{proof}
Let $m$ be any $\sigma$-finite measure on $(\mathcal{H},\sigma(\mathcal{H}))$ with $m_1 \ll m$ and $m_2 \ll m$. 
Let $f = \frac{dm_1}{dm}$ and $g=\frac{dm_2}{dm}$.
Let $h = \min (f,g)$.
We construct $X_n$ and $Y_n$ as follows:
\begin{enumerate}
    \item With probability $c_1=\int_{\mathcal{H}} h dm$: $X_{n+1}=Y_{n+1}  \sim \frac{h}{c_1}$;
    \item With probability $1-c_1$: the two chains proceeds independently as follows:
    
    Let $c_2 = \int_{\mathcal{H}} f-h dm$ and $c_3 = \int_{\mathcal{H}} g-h dm$.
    \begin{enumerate}
        \item $X_{n+1} \sim \frac{f-h}{c_2}$;
        \item $Y_{n+1} \sim \frac{g-h}{c_3}$.
    \end{enumerate}
    Notice that $X_{n+1} \neq Y_{n+1}$.
\end{enumerate}
Thus $X \sim m_1$, $Y \sim m_2$ and
\begin{equation*}
    Prob(X=Y) = c_1 .
\end{equation*}
\textbf{Claim}
\begin{equation*}
    c_1 = 1-\lVert m_1-m_2 \rVert.
\end{equation*}
Let $H = \max (f,g)$.
We rewrite the total variation bound as:
\begin{equation*}
    \lVert m_1-m_2 \rVert = \frac{1}{2} (\int_{f>g} (f-g) dm + \int_{g>f} (g-f) dm) =\frac{1}{2} \int_{\mathcal{H}} (H-h) dm.
\end{equation*}
Since $H+h = f+g$, we have
\begin{equation*}
    \int_{\mathcal{H}} (H+h) dm =2.
\end{equation*}
Hence,
\begin{equation*}
\begin{split}
    \frac{1}{2} \int_{\mathcal{H}} (H-h) dm &= 1- \frac{1}{2}(2- \int_{\mathcal{H}} (H-h) dm)\\
    &= 1- \frac{1}{2}(\int_{\mathcal{H}} (H+h) dm - \int_{\mathcal{H}} (H-h) dm) \\
    &= 1-\int_{\mathcal{H}} h dm \\
    &= 1- c_1
\end{split}
\end{equation*}
The result follows.
\end{proof}

\bigskip
\subsection{The Proof}
\bigskip

\begin{definition}
\label{def2.5}
Suppose the MTMC chain is equipped with the approximation distribution index $A$ and the chain starts with approximation $a_0$ and initial point $\mathbf{x}$. The total variation bound between the distribution of MTMC chain on a state space $\mathcal{H}$ at step $n$ and the target distribution $p(\cdot)$ is given as:
\begin{equation*}
    TV(\mathbf{x}, a_0, n) = \sup_{B\in \sigma(\mathcal{H})} \lvert Prob(X_n \in B| X_0 = \mathbf{x}, A_0 = a_0)-p(B) \rvert, 
\end{equation*}
\end{definition}

\begin{definition}
\label{def2.6}
The MTMC chain is called ergodic if the total variation bound between its distribution at step $n$ and the target distribution $p$ converges to $0$ for all initial values $a_0$ and $\mathbf{x}$, i.e.:
\begin{equation*}
    \lim_{n \rightarrow \infty} TV(\mathbf{x}, a_0, n) = 0
\end{equation*}
\end{definition}

Now we prove the main result of convergence: the MTMC chain is ergodic in the sense of Definition \ref{def2.5}. 

First we give the intuition of the theorem and its proof.

For all $m \in A$, the transition kernel $P_m$ with acceptance ratio (\ref{alpha2}) designed using the approximation distribution $a_m$ has $a_m$ as the stationary distribution (proof later in this section). We require that:
\begin{enumerate}
    \item The variation distance between the transition probability $P_m(\mathbf{x},\cdot)$ and the target distribution $p(\cdot)$ converges to $\delta(m)$ uniformly for all $\mathbf{x} \in \mathcal{H}$ and $m \in A$ where $\delta(m) \rightarrow 0$ for $m \rightarrow \infty$ is the variation distance between $a_m$ and $p$;
    \item The variation distance between each successive transition probabilities, i.e. $P_m(\mathbf{x},\cdot)$ and $P_{m+1}(\mathbf{x},\cdot)$ converges to $0$ in probability for all $\mathbf{x}$.
\end{enumerate}

The first condition merely states that the distribution of a chain with transition kernel $P_m$ converges to $a_m$ and $a_m$ converges to $p$. The second condition gives that the variation distance between each successive transition kernels vanishes almost everywhere over time.

For the proof, we first define an event E, on which the construction of the following two chains is possible:
The first chain is the original chain starting from the $(K-N)$-th iteration and the second chain is a chain with transition kernel $P_{K-N}$ where $K$ and $N$ are chosen carefully so that the probability of these two chains being different at $K$-th iteration is small enough. 

Then, we prove that variation distance on $E$ between the second and the third chain (the chain with transition probability $p$) is small enough.

In total, the probability of the first chain and the third chain not being equal at $K$-th iteration is bounded by the probability of the first and second chain not being equal on $E$ plus the second and third chain not being equal on $E$ plus the probability of $E$ not happening. 

Then we use the general result Theorem \ref{thm2.2} to give the desired ergodicity. 

It is to notice that the convergence does not depend on specific choices of $a_m$ but the way how one successively chooses $a_m$. One can consider the index set $A$ of approximation as a hidden Markov chain. Thus the convergence is universal in the sense that one can choose the approximation method arbitrarily as long as the two constraints above are met.

\begin{theorem}
\label{thm2.6}
Consider MTMC algorithm on a state space $\mathcal{H}$ with approximation index $A$. For each transition kernel $P_m$, the stationary distribution is $a_m$ for all $m \in A$. Assume that for all $\mathbf{x}$:
\begin{itemize}
    \item For all $\epsilon > 0$, there exists $N=N(\epsilon) \in \mathbb{N}$ such that
    \begin{equation}
        \lVert P^N_m (\mathbf{x}, \cdot) - p(\cdot) \rVert \leq \epsilon + \delta(m)
    \end{equation}
    where $\delta(m) \rightarrow 0$ for $m \rightarrow \infty$ and 
    \begin{equation*}
        \delta(m) = \lVert a_m(\cdot)-p(\cdot) \rVert;
    \end{equation*}
    \item \begin{equation*}
        \lim_{m \rightarrow \infty} D_m \rightarrow 0
    \end{equation*}
    where
    \begin{equation*}
        D_m = \sup_{\mathbf{x} \in \mathcal{H}} \lVert P_{m+1}(\mathbf{x},\cdot)-P_{m}(\mathbf{x},\cdot) \rVert
    \end{equation*}
    is a $\mathcal{G}_{m+1} = \sigma (a_0,\dots,a_{m+1},\mathbf{x}^0, \dots, \mathbf{x}^{m+1})$ measurable random variable depending on the random values $m$ and $m+1$. 
\end{itemize}
Then the MTMC chain is ergodic.
\end{theorem}

\begin{proof}
Let $\epsilon >0$. Choose $N=N(\epsilon)$ as in the first condition.
Let $H_m = \{D_m \geq \frac{\epsilon}{N^2} \}$ and choose $m^{*}=m^{*}(\epsilon)\in \mathbb{N}$ large enough such that 
\begin{equation*}
    Prob(H_m) \leq \frac{\epsilon}{N}
\end{equation*}
for $m \geq m^{*}$.
Define the event 
\begin{equation*}
    E := \bigcap_{i=m+1}^{m+N} H_i^c.
\end{equation*}
Then
\begin{equation}
\label{in3}
    Prob(E) \geq 1-\epsilon
\end{equation}
by de Morgan's Law.
By triangle's inequality we have
\begin{equation*}
    \lVert P_{m+k}(\mathbf{x},\cdot)-P_{m}(\mathbf{x},\cdot) \rVert < \frac{1-\epsilon}{N}
\end{equation*}
on $E$ for all $k \leq N$.
Fix some $K \geq m^{*}+N$. Notice that 
\begin{equation}
\label{knt}
    \lVert P_{K-N}(\mathbf{x},\cdot)-P_{t}(\mathbf{x},\cdot) \rVert < \frac{1-\epsilon}{N}
\end{equation}
on $E$ for $K-N \leq t \leq K$.
Now we construct the coupling argument:
First construct the original chain $\{X_n\}$ together with its approximation sequence $\{A_n\}$. Assume it starts with $X_0=\mathbf{x}$ and ${a_0}$.

\textbf{Claim}
We can construct a second chain $\{X^{'}_n\}_{n=K-N}^K$ on $E$ such that:
\begin{enumerate}
    \item $X^{'}_{K-N} = X_{K-N}$;
    \item $X^{'}_n \sim P_{K-N} (X^{'}_{n-1}, \cdot)$ for $K-N +1 \leq n \leq K$;
    \item $Prob(X^{'}_i = X_i, \ for \ K-N \leq i \leq t) \geq 1- (t-(K-N)) \frac{\epsilon}{N}$ for $K-N \leq t \leq K$.
\end{enumerate}
We prove the claim via induction.
\begin{enumerate}
    \item First notice that the claim is trivially true for $t=K-N$;
    \item Suppose it is also true for $t$;
    \item By induction hypothesis we know that $X^{'}_i = X_i$ for all $K-N \leq i \leq t$, $X^{'}_{t+1} \sim P_{K-N} (X^{'}_{t}, \cdot)$ and $X_{t+1} \sim P_t (X_{t}, \cdot)$.
    It follows that the conditional distribution of $X^{'}_{t+1}$ and $X_{t+1}$ are within the range $\frac{\epsilon}{N}$ since we have (\ref{knt}) on $E$.
    Then $X^{'}_{t+1} = X_{t+1}$ with probability $1- \frac{\epsilon}{N}$.
    Thus
    \begin{equation*}
    \begin{split}
        &Prob(X^{'}_i = X_i \ for \ K-N \leq i \leq t+1) \\
        &\geq Prob(X^{'}_i = X_i \ for \ K-N \leq i \leq t)(1-\frac{\epsilon}{N})\\
        &\geq (1- (t-(K-N)) \frac{\epsilon}{N})(1-\frac{\epsilon}{N})\\
        &\geq (1- (t+1-(K-N)) \frac{\epsilon}{N})
    \end{split}
    \end{equation*}
\end{enumerate}
This proves the claim.
In particular, the claim shows that 
\begin{equation*}
\begin{split}
    &Prob(X^{'}_K = X_K) \\
    &\geq 1-(K-(K-N)) \frac{\epsilon}{N}\\
    &= 1-\epsilon.
\end{split}
\end{equation*}
That is:
\begin{equation}
\label{in1}
    Prob(X^{'}_K \neq X_K, E) <  \epsilon.
\end{equation}
Now 
\begin{equation*}
    \lVert P^N_{K-N} (X^{'}_{K-N}, \cdot) - p(\cdot) \rVert \leq \epsilon + \delta(K-N)
\end{equation*}
conditioning on $X_{K-N}$ as in the first assumption.
Integrating of the distribution of $X_{K-N}$ gives:
\begin{equation*}
    \lVert \mathcal{L}(X^{'}_K)-p(\cdot)\rVert < \epsilon + \delta(K-N)
\end{equation*}
We can construct a third chain $Z \sim p(\cdot)$ as showed in Theorem \ref{thm2.3} such that 
\begin{equation}
\label{in2}
    Prob(X^{'}_K \neq Z) < \epsilon + \delta(K-N). 
\end{equation}
So we have for \ref{in1}, \ref{in2} and \ref{in3}:
\begin{equation*}
\begin{split}
    Prob(X_K \neq Z) &\leq Prob(X^{'}_K \neq X_K, E) + Prob(X^{'}_K \neq Z) +Prob(E^c) \\
    &\leq \epsilon + \epsilon + \delta(K-N) + \epsilon \\
    &= 3\epsilon + \delta(K-N).
\end{split}
\end{equation*}
Using Theorem \ref{thm2.2} we have $\lVert \mathcal{L}(X_K) - p(\cdot) \rVert < 3\epsilon + \delta(K-N)$.

Since $m \geq m^{*}$ and $K \geq m^{*} + N$ were arbitrary, $T(\mathbf{x},a_0,K) < 3 \epsilon + \delta(K-N)$ for $K$ large enough. As assumed, We have $\delta$ converges to $0$, thus
\begin{equation*}
    \lim_{n \rightarrow \infty} T(\mathbf{x},a_0,n) =0.
\end{equation*}
\end{proof}

Roberts and Rosenthal \cite{proofofconvergence} used this technique to prove the convergence of an adaptive MCMC. We prove the convergence of MTMC based on its structure.

Now, we consider the 'Moving Target' version of the \emph{Birkhoff–Khinchin theorem}. 

\begin{theorem}(Birkhoff-Khinchin)
Let $(\mathcal{H},\sigma(\mathcal{H}),p)$ be a probability space. Let $e: \mathcal{H} \rightarrow \mathbb{R}$ be a bounded measurable function. Let $P$ be a measure-preserving map. Then with probability $1$
\begin{equation}
\label{BK}
    \lim_{n \rightarrow \infty} \frac{\sum_{k=1}^n e((P)^k \mathbf{x})}{n} =  \int_{\mathcal{H}} e dp.
\end{equation}
\end{theorem}

The left hand side of (\ref{BK}) is the so-called time-average and the right hand side space-average. Intuitively, we can summarise the theorem as: if the transformation kernel is ergodic, and the measure is invariant, then the time average is equal to the space average almost everywhere (convergence in probability). If we are dealing with a Markov chain, the application is immediate. However, with an MTMC chain, more assumptions are needed. In the following, we prove the 'Moving Target' version of the Birkhoff-Khinchin theorem with some coupling constructions used in the previous proof of Theorem \ref{thm2.6}. We use the proof technique based on Roberts' and Rosenthal's proof \cite{proofofconvergence} of adaptive Markov chains.

\begin{proposition}
Suppose we have an MTMC chain with the conditions in Theorem \ref{thm2.6}. Let $e: \mathcal{H} \rightarrow \mathbb{R}$ be a bounded measurable function. Then if the chain starts in $\mathbf{x} \in \mathcal{H}$ and $0 \in A$,
\begin{equation*}
    \frac{\sum_{k=1}^n e(X_k)}{n} \rightarrow \int_{\mathcal{H}} e dp
\end{equation*}
in probability for $n \rightarrow \infty$. 
\end{proposition}

\begin{proof}
Without loss of generality, we take $\int_{\mathcal{H}} e dp=0$. Otherwise add the term $\int_{\mathcal{H}} e dp$ to the right hand side of every equation.

Let $\Tilde{e}=\sup_{\mathbf{x} \in \mathcal{H}} \lvert e(\mathbf{x}) \rvert < \infty$.
Let $\mathbf{E}_{a_m,\mathbf{x}}$ be the expectation with respect to the transition kernel $P_{m}$ when the chain starts from $X_0=\mathbf{x}$. Let $\mathbf{E}$ and $\mathbf{P}$ be the expectations and probabilities with respect to the varying target distributions $a_0, a_1, a_2, \dots$ of the MTMC chain.

Using Theorem \ref{BK}, the first condition of Theorem \ref{thm2.6} implies that given $\epsilon$, we can find a $N=N(\epsilon)$ with
\begin{equation*}
    \mathbf{E}_{a_m,\mathbf{x}} \bigg( \bigg \lvert \frac{\sum_{k=1}^N e(X_k)}{N} \bigg \rvert \bigg) \leq \epsilon + \delta(m)
\end{equation*}
for all $m \in A$ and $\mathbf{x} \in \mathcal{H}$.

Then find $m^{*}$ as in the proof of Theorem \ref{thm2.6} to construct the coupling on $E$ as before. Suppose $\lvert e \rvert$ is bounded by $\Tilde{e}$, we have for (\ref{in1}), (\ref{in2}) and (\ref{in3})
\begin{equation}
\label{ine1}
\begin{split}
    &\mathbf{E}\bigg(\bigg \lvert \frac{\sum_{k=m+1}^{m+N} e(X_k)}{N}\bigg \rvert | \mathcal{G}_m \bigg) \\
    &\leq \mathbf{E}_{a_m,\mathbf{x}} \bigg( \bigg \lvert \frac{\sum_{k=1}^N e(X_k)}{N} \bigg \rvert \bigg) + \Tilde{e} \epsilon + \Tilde{e} \mathbf{P}(E^c) \\
    &= \epsilon + \delta(m) + 2 \Tilde{e} \epsilon
\end{split}
\end{equation}
where $m \geq m^{*}$.
Let $\Tilde{N}$ be large enough such that
\begin{equation}
\label{ine2}
    \max \bigg(\frac{\Tilde{e} m^{*}}{\Tilde{N}}, \frac{\Tilde{e} N}{\Tilde{N}}  \bigg) \leq \epsilon.
\end{equation}
Then
\begin{equation*}
\begin{split}
    \bigg \lvert \frac{\sum_{k=1}^{\Tilde{N}} e(X_k)}{\Tilde{N}}\bigg \rvert \leq& \bigg \lvert \frac{\sum_{k=1}^{\Tilde{m^{*}}} e(X_k)}{\Tilde{N}} \bigg \rvert+
    \bigg \lvert \frac{1}{\lfloor \frac{\Tilde{N}-m^{*}}{N}\rfloor} \sum_{k_1=1}^{\lfloor \frac{\Tilde{N}-m^{*}}{N}\rfloor} \frac{1}{N} \sum_{k_2=1}^N e(X_{m^{*}+(k_1-1)N+k_2})\bigg \rvert + \\
    &\bigg \lvert \frac{\sum_{k=m^{*}+\lfloor \frac{\Tilde{N}-m^{*}}{N}\rfloor N+1 }^{\Tilde{N}} e(X_k)}{\Tilde{N}} \bigg \rvert
\end{split}
\end{equation*}
In this equation, the left hand side has a sum ranging over $[1,\Tilde{N}]$. This sum is decomposed on the right hand side to $[1, m^{*}]$, $[m^{*}+1,m^{*}+\lfloor \frac{\Tilde{N}-m^{*}}{N}\rfloor N]$ and $[m^{*}+\lfloor \frac{\Tilde{N}-m^{*}}{N}\rfloor N, \Tilde{N}]$. Each term is normalised over some constant smaller or equal to $\Tilde{N}$. Thus the right hand side is indeed larger or equal to the left hand side.

The first and third term of the right hand side are both smaller than $\epsilon$ since (\ref{ine2}) and the boundedness of $e$. The second term takes average of the sum in which each term is smaller or equal to $\epsilon + \delta(m) + 2 \Tilde{e} \epsilon$ according to (\ref{ine1}).

Taking expectation on both sides, we have
\begin{equation*}
    \mathbf{E}\bigg(\bigg \lvert \frac{\sum_{k=1}^{\Tilde{N}} e(X_k)}{\Tilde{N}}\bigg \rvert \bigg) \leq \epsilon + \epsilon + \delta(m) + 2 \Tilde{e} \epsilon + \epsilon = 3\epsilon +\delta(m) + 2 \Tilde{e} \epsilon
\end{equation*}
Using Markov's inequality, i.e. $\mathbf{P}(X > c) \leq \frac{\mathbf{E}(X)}{c}$ for some constant $c>0$, we have
\begin{equation*}
    \mathbf{P}\bigg(\bigg \lvert \frac{\sum_{k=1}^{\Tilde{N}} e(X_k)}{\Tilde{N}}\bigg \rvert > \epsilon^{\frac{1}{2}} \bigg) \leq (3+2 \Tilde{e})\epsilon ^{\frac{1}{2}}+\frac{\delta(m)}{\epsilon^{\frac{1}{2}}}
\end{equation*}
Since we can choose $m^{*}$ large enough so that for all $m \geq m^{*}$,  $\frac{\delta (m)}{\epsilon^{\frac{1}{2}}} \rightarrow 0$ for $m \rightarrow \infty$ and $\epsilon$ is arbitrary, the result follows.
\end{proof}

Before we proceed, it is necessary to confirm the first condition of Theorem \ref{thm2.6}. It is sufficient to show that $a_m$ is the (unique) stationary distribution of $P_m$ for all $m \in A$. Notice that when fixing the approximation distribution, the chain is Markovian. 

\begin{lemma}
The approximation distribution $a_m$ is a stationary distribution of the transition kernel $P_m$ for all $m \in A$ where $A$ is the index set for approximation distribution.
\end{lemma}
\begin{proof}
We first show the \emph{detailed balance condition}, i.e.
\begin{equation}
\label{db}
    P_m(\mathbf{x},\mathbf{y})a_m(\mathbf{x})=P_m(\mathbf{y},\mathbf{x})a_m(\mathbf{y})
\end{equation} 

The transition probabilities from $\mathbf{x}$ to $\mathbf{y}$ and from $\mathbf{y}$ to $\mathbf{x}$ are
\begin{equation*}
P_m(\mathbf{x},\mathbf{y})=Q(\mathbf{x},\mathbf{y})\min\bigg(1,\frac{a_m(\mathbf{x})}{a_m(\mathbf{y})}\frac{Q(\mathbf{y},\mathbf{x})}{Q(\mathbf{x},\mathbf{y})}\bigg)
\end{equation*}
and
\begin{equation*}
P_m(\mathbf{y},\mathbf{x})=Q(\mathbf{y},\mathbf{x})\min\bigg(1,\frac{a_m(\mathbf{y})}{a_m(\mathbf{x})}\frac{Q(\mathbf{x},\mathbf{y})}{Q(\mathbf{y},\mathbf{x})}\bigg)
\end{equation*}
Either $\min \Big(1,\frac{a_m(\mathbf{x})}{a_m(\mathbf{y})}\frac{Q(\mathbf{y},\mathbf{x})}{Q(\mathbf{x},\mathbf{y})}\Big)$ or $\min\Big(1,\frac{a_m(\mathbf{y})}{a_m(\mathbf{x})}\frac{Q(\mathbf{x},\mathbf{y})}{Q(\mathbf{y},\mathbf{x})}\Big)$ is equal to $1$.
Without loss of generality, we assume $\min \Big(1,\frac{a_m(\mathbf{y})}{a_m(\mathbf{x})}\frac{Q(\mathbf{x},\mathbf{y})}{Q(\mathbf{y},\mathbf{x})}\Big)=1$. Then
\begin{equation*}
    \frac{P_m(\mathbf{x},\mathbf{y})}{P_m(\mathbf{y},\mathbf{x})}=\frac{Q(\mathbf{x},\mathbf{y})\frac{a_m(\mathbf{x})}{a_m(\mathbf{y})}\frac{Q(\mathbf{y},\mathbf{x})}{Q(\mathbf{x},\mathbf{y})}}{Q(\mathbf{y},\mathbf{x})} = \frac{a_m(\mathbf{x})}{a_m(\mathbf{y})}
\end{equation*}
The detailed balance follows.
If integrate on both sides of (\ref{db}):
\begin{equation}
\label{db2}
\begin{split}
\int_{\mathcal{\mathcal{H}}} P_m(\mathbf{x},\mathbf{y})a_m(\mathbf{x})d\mathbf{x}
&=\int_{\mathcal{H}}P_m(\mathbf{y},\mathbf{x})a_m(\mathbf{y})d{\mathbf{x}}\\
&= a_m(\mathbf{y})\int_{\mathcal{H}} P(\mathbf{y},\mathbf{x})d\mathbf{x}.
\end{split}
\end{equation}
Since $\int_{\mathcal{H}}{P_m(\mathbf{y},\mathbf{x})}d\mathbf{x} = 1$, (\ref{db2}) is equal to $a_m(\mathbf{y})$. The claim follows.
\end{proof}

Now, we show that the stationary distribution $a_m$ is also unique. The proof follows Roberts' and Rosenthal's \cite{couplingconstruct}.

Assume the chain has a positive and continuous proposal density. Further assume that the approximation density is finite on the state space.

\begin{definition}
A chain is called $\psi$-irreducible if there exists a non-zero $\sigma$-finite measure $\psi$ on the state space $\mathcal{H}$ such that for all $B \subseteq \sigma(\mathcal{H})$ with $\psi(B) >0$, and for all $x \in \mathcal{H}$, there exists a positive integer $n=n(x,B)$ such that the $n$-step transition probability $P^n(x,B)>0$.
\end{definition}

\begin{lemma}(Roberts and Rosenthal)\cite{couplingconstruct}
\label{lemma2.10}
The MTMC chain with fixed approximation distribution $a_m$ is $a_m$-irreducible. 
\end{lemma}

$Remark.$ This means that if there exists some set $B$ on the state space with $a_m(B)>0$, the chain could get to $B$ in finitely many steps.

\begin{proof}
Let $\Tilde{a}_m$ be possibly unnormalised density function of $a_m$ with respect to Lebesgue measure. Let $q$ be a continuous and positive proposal density.
    
Define $B_r = B \cap D(r,\mathbf{0})$ where $ D(r,\mathbf{0})$ is an open ball with radius $r$ and centre $\mathbf{0}$.
    
Let $a_m(B)>0$. Then there exists radius $r>0$ such that $a_m(B_r)>0$. For any $\mathbf{x}$ we have for continuity:
\begin{equation*}
\inf_{\mathbf{y} \in B_r} min(q(\mathbf{y},\mathbf{x}),q(\mathbf{x},\mathbf{y})) \geq \epsilon 
\end{equation*}
for some $\epsilon >0$.
Thus 
\begin{equation*}
    \begin{split}
        P_m(\mathbf{x},B) &\geq P_m(\mathbf{x},B_r) \geq \int_{B_r} q(\mathbf{x},\mathbf{y}) \min \bigg(1,\frac{\Tilde{a}_m(\mathbf{y})}{\Tilde{a}_m(\mathbf{x})}\frac{q(\mathbf{y},\mathbf{x})}{q(\mathbf{x},\mathbf{y})}\bigg) d\mathbf{y}\\
        &\geq \epsilon Leb(\left\{\mathbf{y} \in B_r: \Tilde{a}_m(\mathbf{y}) \geq \Tilde{a}_m(\mathbf{x}) \right\} +\\
        &\frac{\epsilon \Tilde{k}}{\Tilde{a}_m(\mathbf{x})} p(\left\{\mathbf{y} \in B_r: \Tilde{a}_m(\mathbf{y}) < \Tilde{a}_m(\mathbf{x})\right\}
    \end{split}
    \end{equation*}
    where 
    \begin{equation*}
        \Tilde{k} = \int_{\mathcal{H}} \Tilde{a}_m(\mathbf{x})d\mathbf{x}>0
    \end{equation*}
    Since $a_m(\cdot)$ is absolutely continuous with respect to Lebesgue measure and $Leb(B_r) > 0$, it follows that the terms cannot be 0, thus  $P_m(\mathbf{x},B) >0$. Hence, the chain is $a_m$-irreducible.
\end{proof}

\begin{definition}
A chain with stationary distribution $\phi$ is called aperiodic if there are no disjoint subsets $\mathcal{H}_1,\cdots,\mathcal{H}_n \subseteq \mathcal{H}$ for $n\geq 2$ with transition probability $P(x,\mathcal{H}_{k+1})=1$ for all $x \in \mathcal{H}_k$ ($1 \leq i \leq n-1$) and $P(x,\mathcal{H}_1)=1$ for all $x \in \mathcal{H}_i$ such that $\phi(\mathcal{H}_1) >0$.
\end{definition}

\begin{lemma}(Roberts and Rosenthal)\cite{couplingconstruct}
The MTMC chain with fixed approximation distribution $a_m$ is aperiodic. 
\end{lemma}

\begin{proof}
Suppose $\mathcal{H}_1$ and $\mathcal{H}_2$ are disjoint subsets of $\mathcal{H}$ with positive $a_m$-measure. Suppose $P(\mathbf{x},\mathcal{H}_2)=1$ for all $\mathbf{x} \in \mathcal{H}_1$.
Take any $\mathbf{x} \in \mathcal{H}_1$, since $\mathcal{H}_1$ have positive Lebesgue measure,
\begin{equation*}
        P(\mathbf{x},\mathcal{H}_1) \geq \int_{\mathbf{y} \in \mathcal{H}_1} q(\mathbf{x},\mathbf{y}) \min \bigg(1, \frac{a_m(\mathbf{y})}{a_m(\mathbf{y})}\frac{q(\mathbf{y},\mathbf{x})}{q(\mathbf{x},\mathbf{y})}\bigg) d\mathbf{y} >0.
\end{equation*}
This is a contradiction.

\end{proof}

\begin{corollary}(Roberts and Rosenthal)\cite{couplingconstruct}
The approximation distribution $a_m$ is the unique stationary distribution of $P_m$ for all $m \in A$.
\end{corollary}

The result follows from the fact that if a Markov chain with stationary distribution $a_m$ is $a_m$-irreducible and aperiodic, then $a_m$ is the unique stationary distribution.
The proof uses the coupling argument can be found in Roberts' and Rosenthal's \cite{couplingconstruct}.

\bigskip
\section{Rate of Convergence}
\bigskip

In this section, we put the 'model error' from corresponding approximation method $\delta$ aside and consider the 'random error' $\epsilon$. More specifically, we examine how fast $\epsilon$ descends when $N$ gets larger when fixing the approximation distribution chosen, i.e. we determine $\epsilon$ for $\lVert P^N_m (\mathbf{x}, \cdot) - p(\cdot) \rVert \leq \epsilon + \delta(m)$ fixing $m$. 

$Notation.$ Since we only consider one transition matrix $P_m$ in this section, we lose the index $m$ and denote the transition probability as $P$ only.

Recall that we denote the approximation distribution chosen at step $m$ as $a_m$, the proposal distribution as $Q$.
On the state space $\mathcal{H}$, the transition probability $P(\mathbf{x},\mathbf{y})$ at each step from $\mathbf{x}$ to $\mathbf{y}$ is
\begin{equation}
P(\mathbf{x},\mathbf{y}) = 
\left\{
             \begin{array}{lr}
             Q(\mathbf{x},\mathbf{y}) \min\Big(1,\frac{a_{m}(\mathbf{y})}{a_{m}(\mathbf{x})}\frac{Q(\mathbf{y},\mathbf{x})}{Q(\mathbf{x},\mathbf{y})}\Big), & \mathbf{x} \neq \mathbf{y}  \\ Q(\mathbf{x},\mathbf{x}) + \sum_{\mathbf{z} \in \mathcal{H}, \mathbf{z} \neq \mathbf{x}} Q(\mathbf{x},\mathbf{z})\max \Big(0, 1-\frac{a_{m}(\mathbf{z})}{a_{m}(\mathbf{x})}\frac{Q(\mathbf{z},\mathbf{x})}{Q(\mathbf{x},\mathbf{z})}\Big), &  \mathbf{x} = \mathbf{y}
             \end{array}
\right.
\label{transprob}
\end{equation} 

Assume that the sample space denoted as $\mathcal{H}=\left\{1,...,n\right\}$ is finite and discrete. 
The transition kernel is then an $n \times n$ right stochastic matrix at each iteration, i.e. a real square matrix with each row summing to $1$:
\begin{gather}\nonumber
\begin{split}
\mathcal{P} &= 
\left[
\begin{matrix}
P(1,1)&P(1,2)&\cdots&P(1,n-1)&P(1,n)\\
P(2,1)&P(2,2)&\cdots&P(2,n-1)&P(2,n)\\
\vdots&\vdots&\ddots&\vdots&\vdots\\
P(n-1,1)&P(n-1,2)&\cdots&P(n-1,n-1)&P(n-1,n)\\
P(n,1)&P(n,2)&\cdots&P(n-1,n)&P(n,n)\\
\end{matrix} \right]
\\
\end{split}
\end{gather}

\bigskip
\subsection{Independent Proposal and Eigenvalue Analysis}
\bigskip

In this section, we assume the proposal $Q$ is independent, i.e. $Q(i,j)=Q(j)=:Q_j$ and use eigenvalue analysis to bound the convergence rate. This technique is quite common when analysing the convergence rate of a Markov chain. The proof's structure and the idea of ordering the state space are based on Liu's \cite{Liu}. 

The transition probability at step $n$ is given as
\begin{equation}\nonumber
    P(i,j) = 
    \left\{
    \begin{array}{lr}
             Q_j \min \Big(1,\frac{a_m(j)}{a_m(i)}\frac{Q_i}{Q_j}\Big), & i \neq j, i,j \in \mathcal{H} \\ Q_i + \sum_{d \in \mathcal{H}, d \neq i} Q_d\max \Big(0, 1-\frac{a_m(d)}{a_m(i)}\frac{Q_i}{Q_d}\Big), &  i = j, i,j \in \mathcal{H}
             \end{array}
     \right.
\end{equation}
where we defined \emph{important ratio} \cite{Liu} $w_{k}$ of state $k$ as
\begin{equation}\nonumber
    w_{k} = \frac{a_{m}(k)}{Q_k}, \qquad \forall k\in\mathcal{H}.
\end{equation}

Assume the sample space is given as $\mathcal{H}=\left\{1,...,n\right\}$ which is sorted according to the magnitudes of their importance ratios:
\begin{equation}
\label{imra}
    w_{1} := \frac{a_{m}(1)}{Q_1} \geq \cdots \geq w_{k} := \frac{a_{m}(k)}{Q_k} \geq \cdots \geq w_{m} := \frac{a_{m}(n)}{Q_n}, \qquad k\in\mathcal{H}.
\end{equation}

The transition probability from $i$ to $j$ with $i,j \in \mathcal{H}$ can then be denoted as

\begin{equation}\nonumber
    P(i,j) = 
    \left\{
    \begin{array}{lr}
             Q_j \min \big(1,\frac{w_{j}}{w_{i}}\big), & i \neq j \\ Q_i + \sum_d Q_d \max \big(0, 1-\frac{w_{d}}{w_{i}}\big), &  i = j
             \end{array}.
     \right.
\end{equation}
With the state space $\mathcal{H}$ sorted as \ref{imra}, we can write:
\begin{equation*}
     P(i,j) = 
    \left\{
    \begin{array}{lr}
              \frac{a_m(j)}{w_{i}}, & i < j \\ Q_i + \sum_d Q_d \max \big(0, 1-\frac{w_{d}}{w_{i}}\big), &  i = j\\ Q_j, & j>i
             \end{array}.
     \right.
\end{equation*}

The transition matrix is then
\begin{gather}\nonumber
\mathcal{P} = \left[
\begin{matrix}
Q_1+\lambda_1 & a_m(2)/w_1 & a_m(3)/w_1 & \cdots & a_m(n-1)/w_1  & a_m(n)/w_1\\
Q_1 & Q_2+\lambda_2 & a_m(3)/w_2 & \cdots & a_m(n-1)/w_2 & a_m(n)/w_2\\
\vdots&\vdots&\vdots&\ddots&\vdots&\vdots\\
Q_1 & Q_2 & Q_3 & \vdots &  Q_{n-1}+\lambda_{n-1}  & a_m(n)/w_{n-1} \\
Q_1 & Q_2 & Q_3 & \vdots &  Q_{n-1} & Q_n \\
\end{matrix}\right]
\end{gather}
where
\begin{equation}
\label{eigenvalue}
    \lambda_k = \sum^n_{d=k} \Big(Q_d-\frac{a_m(d)}{w_k}\Big) = \sum^n_{d=k} \Big(\frac{a_m(d)}{w_d}-\frac{a_m(d)}{w_k}\Big) \qquad \forall k\in\mathcal{H}
\end{equation}
is the probability of being rejected in the next step if the chain is currently at state $k$. This makes sense since the probability of staying in the same state as last step is the probability of choosing the same state with the proposal then accept it (here the probability of accepting the same state after it is chosen with the independent proposal is $1$) plus the probability of choosing some other states with the proposal and then being rejected.

If two states $k$ and $k+1$ have equal importance ratios, then $\lambda_k=\lambda_{k+1}$. The transition matrix is then not diagonalisable. We only consider the case where all $\lambda_k$'s are different.

$Remark.$ This is normally the case in reality since it is rare that two states happen to give exactly the same important ratio. We will only consider the case of the transition matrix being diagonalisable. However, the generalisation to non-diagonalisable transition matrices is trivial using generalised eigenvalues. 

We want to compute the eigenvalues and (left) eigenvectors of $\mathcal{P}$. 

Decompose transition matrix $\mathcal{P}$ as $\mathcal{T} + \mathbf{e}\mathbf{Q}^T$ where $\mathbf{e}=(1,...,1)^T$, $\mathbf{Q}=(Q_1,...,Q_m)^T$. $\mathcal{T}$ is an upper triangle matrix
\begin{gather}\nonumber
    \mathcal{T} =\left[
    \begin{matrix}
            \lambda_1 & \frac{Q_2(w_2-w_1)}{w_1} & \frac{Q_3(w_3-w_1)}{w_1} & \cdots & \frac{Q_{n-1}(w_{n-1}-w_1)}{w_1}  & \frac{Q_{n}(w_n-w_1)}{w_1}\\
0&\lambda_2 & \frac{Q_3(w_3-w_2)}{w_2} & \cdots &\frac{Q_{n-1}(w_{n-1}-w_2)}{w_2}&\frac{Q_n(w_n-w_2)}{w_2}\\
\vdots&\vdots&\vdots&\ddots&\vdots&\vdots\\
0&0&0& \vdots &\lambda_{n-1} & \frac{Q_n(w_n-w_{n-1})}{w_{n-1}}\\
0&0&0& \vdots &0& 0 \\
    \end{matrix}\right].
\end{gather}
So the eigenvalues of $\mathcal{T}$ are $1 \geq \lambda_1\geq \lambda_2\geq ...\geq  \lambda_{n-1}$.

$Remark.$ One computes left eigenvectors here since the rows of the transition matrix sum up to $1$. The transition matrix shall be applied on the right hand side of a probability (row)vector. Thus the eigenvector corresponding to eigenvalue $\lambda_0 = 1$ is the stationary distribution $\mathbf{v}_0 = (a_m(1),\cdots, a_m(n))$.

\begin{lemma}
The eigenvalues and (left) eigenvectors of $\mathcal{T}$ are
$\lambda_k$ and 
\begin{equation*}
    \mathbf{v}_k=(0,\cdots,0,-\sum^n_{d=k+1}a_m(d),a_m(k+1),\cdots,a_m(n))
\end{equation*} 
with $k-1$ zero entries, for $k=1,...,n-1$.
\end{lemma}
\begin{proof}
For all $l<k$, the $l$th element of $\mathbf{v}_k\mathcal{T}$ is zero.
For $l=k$, the $l$th element of $\mathbf{v}_k\mathcal{T}$ is $-\sum^n_{d=k+1}a_m(d)\lambda_k$.
For $l>k$, 

\begin{equation*}
\begin{split}
    (\mathbf{v}_k\mathcal{T})_l 
    =& \Big(-\sum^n_{d=k+1}a_m(d)\Big)\frac{Q_{l}(w_l-w_k)}{w_k} + a_m(k+1)\frac{Q_{l}(w_l-w_{k+1})}{w_{k+1}} + \cdots \\
    &+a_m(l-1)\frac{Q_{l}(w_l-w_{l-1})}{w_{l-1}} + a_m(l)\lambda_l \\
    =& \Big(-\sum^n_{d=k+1}a_m(d)\Big)\Big(\frac{Q_k a_m(l)}{a_m(k)}-Q_l\Big) + a_m(k+1) \Big(\frac{Q_{k+1} a_m(l)}{a_m(k+1)}-Q_l\Big)+\cdots\\
    &+a_m(l-1)\Big(\frac{Q_{l-1} a_m(l)}{a_m(l-1)}-Q_l\Big)+a_m(l)\sum^n_{d=l} \Big(Q_d-\frac{a_m(d)}{w_l}\Big) \\
    =&-\frac{Q_k a_m(l)}{a_m(k)}\sum^n_{d=k+1}a_m(d)+Q_l \sum^n_{d=k+1}a_m(d) + Q_{k+1} a_m(l) - Q_l a_m(k+1) + \cdots \\
    &+Q_{l-1} a_m(l) - Q_l a_m(l-1) + a_m(l)\sum^n_{d=l}Q_d - a_m(l) \frac{Q_l}{a_{m}(l)}\sum^n_{d=l}a_{m}(d)\\
    =&-\frac{Q_k a_m(l)}{a_m(k)}\sum^n_{d=k+1}a_m(d)+Q_l \sum^n_{d=k+1}a_m(d)+a_m(l)\sum^{l-1}_{d=k+1}Q_d-Q_l\sum^{l-1}_{d=k+1}a_m(d)+\\
    &a_m(l)\sum^n_{d=l}Q_d-Q_l\sum^n_{d=l}a_{m}(d)\\
    =&-\frac{Q_k a_m(l)}{a_m(k)}\sum^n_{d=k+1}a_m(d)+Q_l \sum^n_{d=k+1}a_m(d)+a_m(l)\sum^n_{d=k+1}Q_d-Q_l\sum^n_{d=k+1}a_m(d)\\ 
    =&-\frac{Q_k a_n(l)}{a_n(k)}\sum^n_{d=k+1}a_m(d)+a_m(l)\sum^n_{d=k+1}Q_d\\
    =&-\frac{Q_k a_m(l)}{a_m(k)}\sum^n_{d=k+1}a_m(d)-\frac{Q_k a_m(l)}{a_m(k)}a_m(k)+a_m(l) Q_k + a_m(l)\sum^n_{d=k+1}Q_d\\
    =&a_m(l)(\sum^n_{d=k} (Q_d-\frac{Q_k}{a_m(k)}a_m(d)))\\
    =&a_m(l)\lambda_k
\end{split}
\end{equation*}
Thus
\begin{equation}\nonumber
    \mathbf{v}_k\mathcal{T} = \lambda_k \mathbf{v}_k
\end{equation}
\end{proof}

\begin{theorem}
The eigenvalues of the transition matrix $\mathcal{P}$ are $1 \geq \lambda_1\geq \lambda_2\geq ...\geq  \lambda_{n-1}\geq 0$ where $\lambda_k=\sum^n_{d=k} (Q_d-\frac{a_m(d)}{w_k})$. 

The corresponding eigenvectors are $\mathbf{v}_k=(0,\cdots,0,-\sum_{d=k+1}^n a_m(d),a_m(k+1),\cdots,a_m(n))$ with $k-1$ zero entries.
\end{theorem}

\begin{proof}
Since $\mathcal{P}=\mathcal{T} + \mathbf{e}\mathbf{Q}^T$,
\begin{equation}\nonumber
\mathbf{v}_k (\mathcal{T} + \mathbf{e}\mathbf{Q}^T)= \lambda_k \mathbf{v}_k+\mathbf{v}_k(\mathbf{e}\mathbf{Q}^T)
\end{equation}
Also,
\begin{equation*}
\begin{split}
    \mathbf{v}_k \mathbf{e}\mathbf{Q}^T=&(0,\cdots,0,-\sum_{d=k+1}^n a_m(d),a_m(k+1),\cdots,a_m(n)) \cdot\\
    &(1,\cdots,1)^T\cdot(Q(1),\cdots,Q(n))=0
\end{split}
\end{equation*}

\end{proof}

Now, we use the eigenvalues and the eigenvectors to bound the total variance between the distribution at $N$-th step and $a_m$.

\begin{theorem}
Assume the transition matrix has fixed approximation $a_m$. The total variance between the distribution at $N$-th step and $a_m$ is
\begin{equation*}
    \lVert p_N - a_m \rVert \leq \Big( \sum_{k=1}^{n-1} \lvert \theta_k \mathbf{v}_k \rvert \Big) (\lambda_1)^N
\end{equation*}
where $\mathbf{v}_0,\cdots, \mathbf{v}_{n-1}$ are a basis of eigenvectors corresponding to the eigenvalues $\lambda_0=1, \cdots, \lambda_{n-1}$ of the transition matrix $\mathcal{P}$ respectively and $\theta_0, \cdots, \theta_{n-1}$ denotes the renormalising constants:
\begin{equation*}
    p_0 = \theta_0 \mathbf{v}_0 + \cdots + \theta_{n-1} \mathbf{v}_{n-1}.
\end{equation*}
where $p_0$ is the initial distribution.
\end{theorem}
\begin{proof}
Since we have the eigenvalues and eigenvector of the trnasition matrix $P$ with $\mathbf{v}_k P = \lambda_k \mathbf{v}_k$ for all $0 \leq k \leq m-1$, we have
\begin{equation*}
\begin{split}
    p_N &= p_0 \mathcal{P}^N\\
    &= \theta_0 \mathbf{v}_0 \mathcal{P}^N+ \cdots + \theta_{n-1} \mathbf{v}_{n-1}\mathcal{P}^N \\
    &= \theta_0 (\lambda_0)^N \mathbf{v}_0 + \cdots + \theta_{n-1}(\lambda_{n-1})^N \mathbf{v}_{n-1}
    \end{split}
\end{equation*}
All the eigenvalues except for $\lambda_0$ is smaller than $1$ and hence all terms except the first vanish when $N \rightarrow \infty$. Thus $p_N \rightarrow \theta_0 \mathbf{v}_0$ for $N \rightarrow \infty$. It is clear that $a_m=\theta_0 \mathbf{v}_0$ is the target distribution up to a constant $\theta_0 = \sum_{i \in \mathcal{H}}\mathbf{v}_0(i)$ which is independent of the initial distribution $p_0$.

Thus 
\begin{equation*}
    p_N(i)-a_m(i) = \theta_1 (\lambda_1)^N \mathbf{v}_1(i)+ \cdots + \theta_{n-1} (\lambda_{n-1})^N \mathbf{v}_{n-1}(i).
\end{equation*}
Then we have
\begin{equation*}
    \lvert p_N(i)-a_m(i) \rvert \leq \Big( \sum_{k=1}^{n-1} \lvert \theta_k \mathbf{v}_k (i) \rvert \lvert \lambda_k\rvert^N \Big) 
\end{equation*}
by the triangle inequality.
Since $\lambda_1$ is the largest eigenvalue smaller than $1$ and total variance take supremum over all $i$, the claim follows.
\end{proof}

$Remark.$ The theorem gives a explicit relation on $N$ and $\epsilon$ for independent proposal. The $\epsilon$ in \ref{thm2.6} is proportional to $(\lambda_1)^N$. 

\begin{proposition}
The bound in Theorem \ref{thm2.6} is $\max(\mathcal{O}((\lambda_1)^N),\mathcal{O}(\delta(N))$ for independent proposal distribution.

\end{proposition}

\bigskip
\subsection{General Case and Coupling Argument}
\bigskip

Now, we introduce the \emph{minorisation condition} for the MTMC chain so that we can bound the convergence rate using the coupling inequality. Again, we fix the index of approximation distribution to bound $\epsilon$.
\begin{definition}
A Markov chain with minorisation condition satisfies an inequality of the form
\begin{equation*}
    P^{N_0}(x,B) \geq \epsilon \gamma (B), \qquad \forall x \in R, \quad \forall t, \quad \forall B \subseteq \sigma(\mathcal{H})
\end{equation*}
where $N_0$ is a positive integer, $R$ is a subset of the state space $\mathcal{H}$, $\epsilon >0$ and $\gamma(\cdot)$ is some probability distribution on $(\mathcal{H}, \sigma(\mathcal{H}))$.
\end{definition}

We prove the general case where $R$ is not the whole state space. The proof follows Rosenthal's \cite{min}.

\begin{theorem}(Rosenthal)
Suppose that an MTMC chain with approximation distribution $a_m$ fixed satisfies the minorisation condition as defined above. Let $(X_k)$ and $(Y_k)$ be two realisations of the MTMC chain with the same transition probability at each step but different initial distribution as described in this section earlier. Let 
\begin{equation*}
    t_1 = \inf \left\{k: (X_k, Y_k) \in R \times R\right\}.
\end{equation*}
and for $i>1$:
\begin{equation*}
    t_i = \inf \left\{k: k \geq t_{i-1} + N_0, (X_k, Y_k) \in R \times R\right\}.
\end{equation*}
Set $z_N = \max \left\{i:t_i<N\right\}$. Then for any $j>0$,
\begin{equation*}
    \rVert \mathcal{L}(X_k) - a_m(\cdot) \rVert = \rVert \mathcal{L}(X_k) - \mathcal{L}(Y_k) \rVert \leq (1-\epsilon)^{\lceil \frac{j}{N_0}\rceil} + Prob(z_N <j).
\end{equation*}
\end{theorem}

\begin{proof}
Without loss of generality, we take $N_0 = 1$ since the variation distance to a stationary distribution is decreasing with $N$ getting larger. We construct the chains $(X_k)$ and $(Y_k)$ as follows:
\begin{enumerate}
    \item Set $X_0$ and $Y_0$ as the initial distribution $p_0$ and the target distribution $a_m$.
    \item For each step $k$, if $X_k$ and $Y_k$ are both in $R$:
    \begin{enumerate}
        \item \label{1a} With probability $\epsilon$, we set $X_{k+1}=Y_{k+1}=\mathbf{x}$ with $\mathbf{x}$ some point in the state space according to distribution $\gamma$;
        \item With probability $1-\epsilon$, we choose $X_{k+1}$ and $Y_{k+1}$ independently according to the distribution $\frac{1}{1-\epsilon}(P_m(X_k,\cdot)- \epsilon \gamma(\cdot))$ and $\frac{1}{1-\epsilon}(P_m(Y_k,\cdot)- \epsilon \gamma(\cdot))$, respectively. 
    \end{enumerate}
    \item If $X_k \notin R$ or $Y_k \notin R$, we choose $X_{k+1}$ and $Y_{k+1}$ independently according to $P_m(X_k, \cdot)$ and $P_m(Y_k, \cdot)$ respectively.
\end{enumerate}
It is straightforward that the chains proceed with the transition probability $P_m(x,\cdot)$. Furthermore, define $T$ as the first time $(X_k)$ and $(Y_k)$ are coupled, that is, the first time that the situation \ref{1a} happens. Now the coupling inequality shows
\begin{equation*}
    \rVert \mathcal{L}(X_k) - a_m(\cdot) \rVert \leq Prob(X_k \neq Y_k) \leq Prob(T >k).
\end{equation*}
Conditional on $X_k$ and $Y_k$ both remaining in $R$, the coupling time $T$ will be a geometric random variable with parameter $\epsilon$. 
Since 
\begin{equation*}
    z_N = \# \left\{ k^{'}<N: (X_{k^{'}},Y_{k^{'}}) \in R \times R \right\},
\end{equation*}
we have
\begin{equation*}
    Prob(T > N, z_N\geq j) \leq (1-\epsilon)^j.
\end{equation*}
Thus
\begin{equation*}
    Prob(T > N) \leq (1- \epsilon)^j + Prob(z_N < j).
\end{equation*}

\end{proof}

In the case where we have the \emph{Doeblin condition}, i.e. $R=\mathcal{H}$, the minorisation condition is vaid for all the $\mathbf{x}$ in the state space. Then $Prob(z_N<N)=0$ and we obtain the following proposition since the second term of the right hand side vanishes.

\begin{proposition}
Suppose that an MTMC chain satisfies the minorisation condition on the entire state space, then
\begin{equation*}
    \rVert P_m^N - a_m(\cdot) \rVert \leq (1-\epsilon)^{\lceil \frac{N}{N_0}\rceil}.
\end{equation*}
\end{proposition}

Recall the set up in the proof of Lemma \ref{lemma2.10}. 

\begin{proposition}(Rosenthal)\cite{couplingconstruct}
Assume the MTMC chain has fixed approximation distribution $a_m$. Assume the state space $\mathcal{H}$ is compact. Assume the approximation density $\Tilde{a}_m$ is continuous and finite. Further assume the proposal density $q$ is positive and continuous. Then $R=\mathcal{H}$ and
\begin{equation*}
    \rVert P_m^N - a_m(\cdot) \rVert \leq (1-\epsilon)^{\lceil \frac{N}{N_0}\rceil}.
\end{equation*}
\end{proposition}

\begin{proof}
Suppose $C$ is a compact set on which $\Tilde{a}_m$ is bounded by some constant $c$. Let $\mathbf{x} \in C$ and $\mathbf{y} \in D$ such that
\begin{equation*}
    \min (q(\mathbf{x},\mathbf{y}),q(\mathbf{y},\mathbf{x})) = \epsilon .
\end{equation*}
Then
\begin{equation*}
    P_m(\mathbf{x},d\mathbf{y}) \geq q(\mathbf{x},\mathbf{y}) \min \bigg( 1, \frac{\Tilde{a}_m(\mathbf{y})}{\Tilde{a}_m(\mathbf{x})} \frac{q(\mathbf{y},\mathbf{x})}{q(\mathbf{x},\mathbf{y})} \bigg)d\mathbf{y} \geq \epsilon \min \bigg(1,\frac{\Tilde{a}_m(\mathbf{y})}{c}\bigg)d\mathbf{y}
\end{equation*}
independent of $\mathbf{x}$. The claim follows.
\end{proof}

$Remark.$ If the state space is finite as in the last subsection, the Doeblin condition is trivially true.

In this subsection, we bound the rate of convergence by bounding $\epsilon$ in the first condition of Theorem \ref{thm2.6}. This is possible since the MTMC chain is Markovian when the acceptance ratio (\ref{alpha2}) is unchanged from iteration to iteration by fixing $a_m$ as approximation distribution. On the other hand, giving explicit values of $\delta(m)$ in the assumptions of Theorem \ref{thm2.6} is possible when knowing the approximation method explicitly, i.e. how the next $a_{m+1}$ is chosen based on the history $\mathcal{G}_{m} = \sigma (a_0,\dots,a_{m},\mathbf{x}^0, \dots, \mathbf{x}^{m})$.

\section{Discussion}
In this paper, we propose the MTMC algorithm and investigate its validity under certain constraints. For any approximation method satisfying the conditions in Theorem \ref{thm2.6}, the convergence is guaranteed and the rate of convergence can be estimated based on the general $\epsilon$ and the $\delta$ of the method chosen. It is natural to consider what kind of approximation method should be used for a 'fast' convergence. The simplest choices include different kinds of interpolations or regressions. As long as the two conditions in the Theorem \ref{thm2.6} are fulfilled, we shall not limit ourselves with one approximation method throughout. For instance, the scheme where one admits different methods of approximation in different parts of the state space can be considered. One could also consider an adaptive approximation method depending on certain variance distances by adding extra terms in the algorithm deciding which approximation method to be used in the next iteration. Currently, Our ongoing work is studying these and some related ideas based on numerical simulations.

However, the reader shall always keep the No-free-lunch Theorem in mind:

\textit{"A general-purpose universal optimization strategy is theoretically impossible, and the only way one strategy can outperform another is if it is specialized to the specific problem under consideration".}\cite{NFL}

No approximation method performs better than the others when chosen without prior knowledge. Thus, it remains crucial to know what kind of problem one is dealing with when choosing the approximation method (or the adaptive method which chooses the approximation method iteratively). 